\documentclass[runningheads]{llncs}

\usepackage{fancyhdr}
\usepackage{cases}
\usepackage{stfloats}
\usepackage{enumerate}
\usepackage{url}
\usepackage{multicol}
\usepackage{graphicx}
\usepackage{wrapfig}
\usepackage{picinpar}
\usepackage{cutwin}
\usepackage{picins}
\usepackage{balance}
\usepackage{algorithm}
\usepackage{algpseudocode}

 \usepackage{etoolbox}

\begin{document}
 
\title{Most memory efficient distributed super points detection on core networks}

\author{Jie Xu\inst{1}  \and
Wei Ding\inst{2}  \and
Xiaoyan Hu\inst{2} }
\authorrunning{J. Xu et al.}
 
\institute{School of Computer Science and Engineering, Southeast University, Nanjing, China \email{xujieip@163.com} \and
School of Cyber Science and Engineering, Southeast University, Nanjing, China \email{wding@njnet.edu.cn}}

\maketitle

\begin{abstract}
The super point, a host which communicates with lots of others, is a kind of special hosts gotten great focus. Mining super point at the edge of a network is the foundation of many network research fields. In this paper, we proposed the most memory efficient super points detection scheme. This scheme contains a super points reconstruction algorithm called short estimator and a super points filter algorithm called long estimator. Short estimator gives a super points candidate list using thousands of bytes memory and long estimator improves the accuracy of detection result using millions of bytes memory. Combining short estimator and long estimator, our scheme acquires the highest accuracy using the smallest memory than other algorithms. There is no data conflict and floating operation in our scheme. This ensures that our scheme is suitable for parallel running and we deploy our scheme on a common GPU to accelerate processing speed. We also describe how to extend our algorithm to sliding time. Experiments on several real-world core network traffics show that our algorithm acquires the highest accuracy with only consuming littler than one-fifth memory of other algorithms.
\end{abstract}

\section{Introduction}
With the developing of the network, thousands of Gigabytes data pass through the Internet every second \cite{cisco:NetForcast}. It is too expensive to monitor every host in the network. An efficient way is to focus on special ones which have great influence on the network security and management. The super point, a host which communicates with lots of others, is one of such special hosts playing important roles in the network, such as Web servers\cite{ICPADS2011:AFailureDetectionServiceForInternetBasedMultiASDistributedSystems}\cite{IMC:BrowserFeatureUsageModernWeb}, P2P spreaders\cite{HSD:identifyHighCardinalityHosts}\cite{P2P:SoKP2PWNED}, DDoS victims\cite{DosC:AttackProtectionCloudComputingSoftwareDefinedNetworking}\cite{DosC:AgainstDataCenterWithCorrelationAnalysis}, scanners\cite{Scan:EvasionResistantNetworkScanDetection}\cite{scan:surveyPortScansAndDetection} and so on. Detecting super point can help us with network management and security. It is also a foundation module of many instruction detection system\cite{Secure:SnortLightweightIntrusionDetectionNetworks}.

For example, DDoS (Distributed Denial of Service) attack is a heavy threat to the Internet\cite{ACSAC:2014:DoS:CPSDrivingCyberphysicalSystems}\cite{IMC:MeasuringtheAdoptionofDDoSProtectionServices}. It appears at the beginning of the Internet and becomes complex with the rapid growth of the network technology.  Although many defense algorithms have been proposed, most of them are too elaborate to deploy in the high-speed network. The peculiarity of a victim under DDoS attack is that it will receive huge packets with different source IP addresses in a short period. A DDoS victim is a typical super point.  Super point only accounts for a small fraction of the overall hosts. If we detect super points first and spend more monitoring resource to them, we can defense DDoS much more efficiently.  Real-time super points detection on the core network is an important step of these applications. 

The speed of nowadays network is growing rapidly. For a core network, it always contains several border routers which locate at different places. How to detect overall super points from all of these distributed routers is more difficult than from a small single router. A distributed super points detection algorithm should satisfy following criteria:
\begin{enumerate}
\item High accuracy.
\item Small memory requirement.
\item Real-time packets processing time.
\end{enumerate}
A high accuracy algorithm should detect out all super points and does not report normal hosts as super points by mistake. Many researchers try to use small and fast memory, such as static random accessing memory SRAM\cite{HSD:DetectionSuperpointsVectorBloomFilter}\cite{HSD:IdentifyHighCardinalitHostNetworkWideTrafficMeasurement}, to detect super point. These algorithms used estimating method to record hosts' cardinalities, the opposite hosts number during a time period. But the accuracy of these algorithms will decrease with the reduction of memory. Parallel computation ability of GPU (Graphic Processing Unit) is stronger than that of CPU because of its plenty operating cores. When using GPU to scan packets parallel, we would get a high throughput and that is what we do in this paper.

To overcome previous algorithms' weakness, we devise a novel distributed super points detection algorithm which has the highest accuracy but consumes smaller than one-fifth memory used by other algorithms. The contributions of this paper are listed following:
\begin{enumerate}
\item A tiny super points detection algorithm is proposed in this paper.  
\item A high accuracy super points filtering algorithm is proposed.  
\item We design the most memory efficient scheme for distributed super points detection.  
\item We extend our algorithm to sliding time window by adopting a more powerful counter. 
\item We implement our algorithm on GPU for real-time super point detection on core network traffic.  
\end{enumerate}

In the next section, we will introduce other super point detection algorithms and analyze their merit and weakness. In section 3, our novel memory efficient algorithm will be represented in detail. How to deploy our algorithm in GPU is described in section 4. In section \ref{sec-slidingWindow}, we describe how to run our algorithm under sliding time window. Section \ref{sec-experiments} shows experiments of our algorithm compared with other ones. And we make a conclusion in the last section. 

\section{Related work}
Super point detection is a hot topic in network research field. Shobha et al.\cite{HSD:streamingAlgorithmFastDetectionSuperspreaders} proposed an algorithm that did not keep the state of every host so this algorithm can scale very well. Cao et al.\cite{HSD:identifyHighCardinalityHosts} used a pair-based sampling method to eliminate the majority of low opposite number hosts and reserved more resource to estimate the opposite number of the resting hosts. Estan et al.\cite{HSD:bitmapCountingActiveFlowsHighSpeedLinks} proposed two bits map algorithms based on sampling flows. Several hosts could share a bit of this map to reduce memory consumption. All of these methods were based on sampling flows which limited their accuracy. 

In these previous algorithms, only a few were suitable for running in distributed environment\cite{HSD:ADataStreamingMethodMonitorHostConnectionDegreeHighSpeed}\cite{HSD:DetectionSuperpointsVectorBloomFilter}\cite{HSD:GPU:2014:AGrandSpreadEstimatorUsingGPU}.

Wang et al.\cite{HSD:ADataStreamingMethodMonitorHostConnectionDegreeHighSpeed} devised a novel structure, called double connection degree sketch (DCDS), to store and estimate different hosts cardinalities. They updated DCDS by setting several bits to one simply. In order to restore super points at the end of a time period, which bits to be updated were determined by Chinese Remainder Theory(CRT) when parsing a packet. By using CRT, every bit of DCDS could be shared by different hosts. But the computing process of CRT was very complex which limited the speed of this algorithm.

Liu et al.\cite{HSD:DetectionSuperpointsVectorBloomFilter} proposed a simple method to restore super hosts basing on bloom filter. They called this algorithm as Vector Bloom Filter(VBF). VBF used the bits extracted from IP address to decide which bits to be updated when scanning a packet. Compared with CRT, bit extraction only needed a small operation.  But VBF would consume much time to restore super point when the number of super points was very big because it used four bit arrays to record cardinalities.

Most of the previous works only focused on accelerating speed by adapting fast memory but neglected the calculation ability of processors. Seon-Ho et al.\cite{HSD:GPU:2014:AGrandSpreadEstimatorUsingGPU} first used GPU to estimate hosts opposite numbers. They devised a Collision-tolerant hash table to filter flows from origin traffic and used a bitmap data structure to record and estimate hosts' opposite numbers. But this method needed to store IP address of every flow while scanning traffic because they could not restore super points from the bitmap directly. Additional candidate IP address storing space increased the memory requirement of this algorithm.

To reduce transmission data in the distributed environment, we devise a novel super point opposite number estimator which can tell if a host is a super point with only 8 bits. Base on this memory efficient estimator, a smart super point restoring algorithm is devised. We will describe our novel algorithm in the following section.
\section{Super point detection} 
The Super point is a host which contacts with many others in a time period $T$. ``Other host" here has different means under different cases. When monitoring opposite IP at a host's network card, other host means every one that sends packets to or receives packets from this host. But this kind of opposite IPs could only be counted by each host self. Generally, a host is locating in a network managed by some Internet Service Providers (ISP). The managers of this subnetwork hope to get information about the traffic between their network and others. From the inspection of ISP, opposite host represents one that being watched at the edge of ISP's subnetwork. Edge of a network means a set of routers forwarding packets between this network and other networks. When monitoring traffic between different network, a router could be regarded as a watch point(WP). 

Let SNet represent the subnetwork managed by an ISP and ONet represent the set of other network communicating with SNet through its edge routers. When detecting the super point at the edge of SNet, the set of a host's opposite IP addresses is defined as below.
\begin{definition}[Opposite IP set/ Opposite IP number]
\label{def-oppositeIPset}
For a host $h$ in SNet or ONet, its Opposite IP set is the set of IP addresses communicating with it over a certain time period $T$ through the edge of SNet written as $OP(h)$. $h$'s opposite IP number is the number of elements in $OP(h)$ denoted as $|OP(h)|$.
\end{definition}

Then we can give the definition of the super point used in this paper.
\begin{definition}[Super point]
\label{def-superPoint}
In a certain time period T, if a host $h$ in $SNet$ or $ONet$ has no less than $\theta $ opposite IPs, $|OP(h)| \geq \theta$, $h$ is a super point.
\end{definition}

Super points may be located in $SNet$ or $ONet$. Both of these super points could be found out by the same algorithm with changing the order of IP addresses. In the rest of this paper, super point means SNet's super point briefly. Opposite IP number estimation is the foundation of super point detection. In this paper, we proposed two novel estimators: short estimator and long estimator.
\subsection{Short Estimator}
In order to judge if a host is a super point, we should record its opposite IP addresses while scanning packets sending to it or it receives. The estimation accuracy is related to the size of memory using to record opposite number. The bigger the size of allocating memory, the more accuracy the result will be. One of the most memory efficient algorithms is OPT proposed by Daniel et al. \cite{DC:AnOptimalAlgorithmDistinctElementProblem}. But the computing complex of OPT is very complex. In this paper, we proposed a more memory efficient algorithm derived from OPT to judge if the opposite IP's number is more than a threshold $\theta$. We call this method as Short Estimator(SE) because it uses only 8 bits, much shorter than other algorithms consumed.

Suppose that SE consists of $g$ bits. 8 is big enough for $g$ when host's IP address is version 4. Every bit of SE is initialized to 0 at the begin of a time period. When receiving a packet related to host $h$, $h$'s opposite IP address $oip$ in this packet will update a bit of SE if the least significant bit of $oip$'s randomly hashed value is bigger than an integer $\tau$. $oip$ is firstly hashed by a random hash function\cite{hash_UniversalClassesOfHashFunctions} $H_1$ to make sure that the hashed value is uniform distribution. $H_1$ will map an integer between 0 and $2^{32}-1$ to another random value in the same range.

If $lsb(H_1(oip)) \geq \tau$, one bit in SE, seleted by another hash function $H_2$, will be set to 1. $H_2$ map an integer between 0 and $2^{32}-1$ to a random value between 0 and $g-1$.

$\tau$ is an integer derived from $\theta$ by the following equation:
\begin{equation}\label{getlsbthreshold}
\tau=ceil(log_2(\theta/8))
\end{equation} 
Function ceil(x) returns the smallest integer no less than x. After scanning all $h$'s relevant packets, we can judge if $|OP(h)|$ is bigger than $\theta$ by counting the number of ``1" bits in SE. The number of ``1" bits in SE is also called the weight of SE, written as $|SE|$. If $|SE| \geq 3$, $|OP(h)|$ is judged as bigger than $\theta$.

We have introduced how to judge if a single host is a super point by SE. But there are millions of host in the network and it's not reasonable to allocate a SE for each of them because of the following two reasons:
\begin{enumerate}
\item Too many memory is required. A core network always contains millions of host. For an IP address of version 4, it will consume 4 bytes. Together with 8 bits used by an SE, each host requires 5 bytes. For a core network containing 100 millions of hosts, we will allocate more than 500 millions of bytes. Such big size of memory is a heavy burden for both memory allocation on server and transmission in the distributed environment.

\item The memory location is very difficult for huge hosts. IP addresses of hosts are widely distributed between 0 and $2^{32}-1$, especially for IP addresses of ONet. How to store and access these randomly hosts efficiently is a hard task. No matter storing these IP addresses in a list or hash table, we have to spend much time in memory accession.
\end{enumerate}
To overcome previous questions, we design an SE sharing structure which can use a fixed number of SE to judge and restore lots of hosts. In the next section, we will introduce super point restoring algorithm based on SE.
\subsection{Restoring super points by short estimator}
Without allocating an SE for each host, we won't know the IP address of super point. In order to detect super points at the end of a time period, we have to reconstruction IP address from our data structure. This requires that our data structure will contain enough IP address information when updating. Based on this requirement, we design a novel structure called Short Estimator Array, written as SEA, which can avoid keeping huge SE instances but can restore super point easily. 

From the name of the SEA, we can see that it is an array of SE. SEA has $SR$ rows and the $i$th row contains $SC(i)$ SEs. When receiving a packet, we extract an IP pair with format $<hip,oip>$ from it where $hip$ is the IP address of the host that we want to monitor and $oip$ is its opposite IP address. An SE in each row of the SEA will record $oip$. $hip$ decides which SE of each row is chosen. Because there are huge hosts in the network, using a single SEA would cause that its SEs is overshared by many hosts. ``Overshare" means there are too many hosts map to the same SE. So we use $2^r$ SEAs and each SEA record a part of traffic. We call these $2^r$ SEAs as SEA Vector (SEAV). Using which SEA to record and estimate a host's opposite IP number is decided by the rightest $r$ bits of the host. We call the right r bits of a host as Right Part (RP) and the rest left 32-r bits as the Left Part (LP). Figure \ref{SEAVpacket} shows how to choose SEA to record opposite IP number.
\begin{figure}[!ht]
\centering
\includegraphics[width=0.77\textwidth]{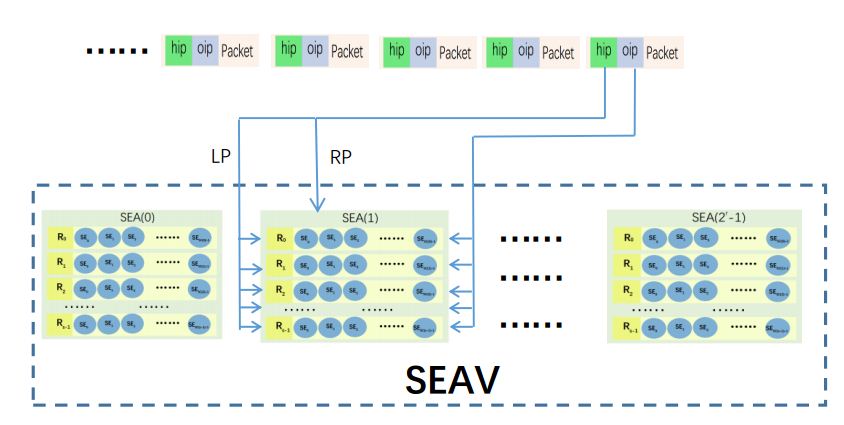}
\caption{Short estimarot arrays vector}
\label{SEAVpacket}
\end{figure}

For a certain SEA in SEAV, RP of a host is clear and only LP is unknown. When choosing a SE to update, we extract several successive bits from $hip$'s LP as the index of SE in each row. SE's index, written as $Idx(i,LP(h))$, means the $i$th row's column identifier of SE relating to a host $h$ whose LP is $LP(h)$. If every bit of LP is contained in one or more SEs' indexes, we can restore LP by extracting and concatenating bits in all these indexes. 
To explain how to get $Idx(i,lp)$ of each row, we give the following declaration:
\begin{definition}[index starting bit]
\label{def-IndexStartBit}
 For the $i$th row in a SEA, its index starting bit $ISB(i)$ is a bit in the LP from which we begin to extract sub bits of the LP.
\end{definition}

\begin{definition}[Index Bits Number]
\label{def-IndexBitsNumber}
 For the $i$th row in a SEA, its index bits number $IBN(i)$ is the number of bits that we would extract from the LP.
\end{definition}

$IBN(i)$ decides the number of $SE$, which is written as $SC(i)$, in the $i$th row. We can acquire $Idx(i,LP(h))$ by $IBN(i)$ and $ISB(i)$. Let $Idx(i,LP(h))[j]$ point to the $j$th bit of $Idx(i,LP(h))$ and $LP(h)[j]$ represent the $j$th bit of $LP(h)$. Every bit of $Idx(i,LP(h))$ could be determined by the following equation.
\begin{equation}\label{eqt_IdxLP}
Idx(i,LP(h))[j]=LP(h)[ (ISB[i]+j) mod (32-r)]
\end{equation} 
Where $0\leq j \leq IBN(i)-1$ and $0 \leq i \leq SR-1$.
The value of $IBN(i)$ and $ISB(i)$ should obey two constraints.
\begin{enumerate}
\item Every bit of $LP(h)$ locates in at least one index. This constraint makes sure that $LP(h)$ could be reconstructed from the $SR$ indexes. In another word, for $j \in [0,31-r]$, there is at least a $i \in [0, SR-1]$ that let $ISB[i] \leq j \leq ISB[i]+IBN[i]-1 $.
\item One index should have several bits same to its next index part. These duplicating bits could help us to remove fake candidate IP addresses efficiently. For $i \in [0, SR-1]$, $(ISB[i]+IBN[i]) mod (32-r) \leq ISB[(i+1)mod SR]$.
\end{enumerate}

The rows number $SR$ will affect $ISB(i)$ and $IBN(i)$. When $SR$ choose a bigger number, the first condition could be matched even all $IBN(i)$s are small. For example, we can set the value of $IBN(i)$ equal to $ceil(\frac{32}{SR}) + a$  and $ISB(i)=i*ceil(\frac{32}{SR})$, where $i\in [0,SR-1]$ and $a$ is a positive integer. When $a$ is fixed, $IBN(i)$ will decrease with the increasing of $SR$. Small $IBN(i)$ causes small memory consumption of SEA. Because each packet will be updated by $SR$ $SE$s, when $SR$ is very big, the updating time will increase too. Considering that the memory requirement of a single SE is very small, only one byte, we can set $SR$ to a small value such as 3 or 4.   

At the end of a time period, super points will be restored from SEVA. By equation \ref{eqt_IdxLP}, we can see that each bit of $LP(h)$ could be recovered from the set of $Idx(LP(h))=\{Idx(i,LP(h))| 0 \leq i \leq SR-1\}$ by a reverse equation as shown below.
\begin{equation}\label{eqt_LP_from_Idx}
LP(h)[j]=Idx(i,LP(h))[j-ISB(i)]
\end{equation} 
Where $0 \leq j \leq 31$, $0\leq i \leq SR-1$, $ISB(i)\leq j$ and $ j-ISB(i) \leq IBN(i)$.

But when update SEAV, we do not record $Idx(LP(h))$. We can derive that if a host $h$ is a super point, $SE(i, Idx(i,LP(h))$ will contain no less than 3 `1' bits. We call these SE whose weight is no less than 3 as Hot SE (HSE). Let $HSE(i)$ represent the set of HSE in the $i$th row. 

By picking $SR$ HSEs from every $HSE(i)$, we can get a candidate index set $CIdx=<c[0],c[1],\cdots,c[SR-1]>$ where $ SE(i, c(i))\in HSE(i)$ and $0\leq i\leq SR-1$. Supposing there are $|HSE(i)|$ elements in the $i$th row, there would be total $\prod_{i=0}^{SR-1}|HSE(i)|$ $CIdx$. The set of all $CIdx$ is denoted by $CIS$. If $h$ is a super point, $Idx(LP(h)) \in CIS$. By test each $CIdx$ in $CIS$ we can reconstruct all super points. 

In order to reduce this influence, we test if the union $SEU$, acquiring by the bit-wise ``AND" of all these $SR$ $SEs$, still contains more than $\theta$ opposite hosts. When its weight is no less than 3, we will return the restored LP as a super point's LP. Together with RP, the index of this SEAV, a super point will be restored. But some normal hosts, whose opposite number is littler than $\theta$ will stay in the result too. To reduce the number of these normal hosts, we apply a more precise estimation method, long estimator, together with the short estimator.
\subsection{Long estimator}
Long estimator are used to improve the accuracy of detection result. It uses more bits to estimate host's opposite number. Our long estimator is based on linear distinct counting algorithm (LDC) \cite{DC:aLinearTimeProbabilisticCountingDatabaseApp}. LDC is a bit vector of $k$ bits initialized with 0. When recording an opposite host $oip$, one bit in LDC, chosen by a random hash function $H_3(oip)$, will be set. LDC has the simplest updating process. The opposite hosts number could be estimated by the following equation:
\begin{equation}\label{eqt_LDC_estValueFromZeroBitsN}
 {Est}'=-k*ln(\frac{z_0}{k})
\end{equation}
$z_0$ is the resting zero number in LDC after recording all opposite hosts. LDC has a good accuracy performance in estimating opposite host number, but its memory consumption is very large. So it's too expensive to allocate an LDC for every host. Like SEA, we construct an array of LDC, written as $LDCA$, with $LR$ rows and $LC$ columns. Figure \ref{LDCAstruct} describes the structure of $LDCA$ and how to update it. 
\begin{figure}[!ht]
\centering
\includegraphics[width=0.77\textwidth]{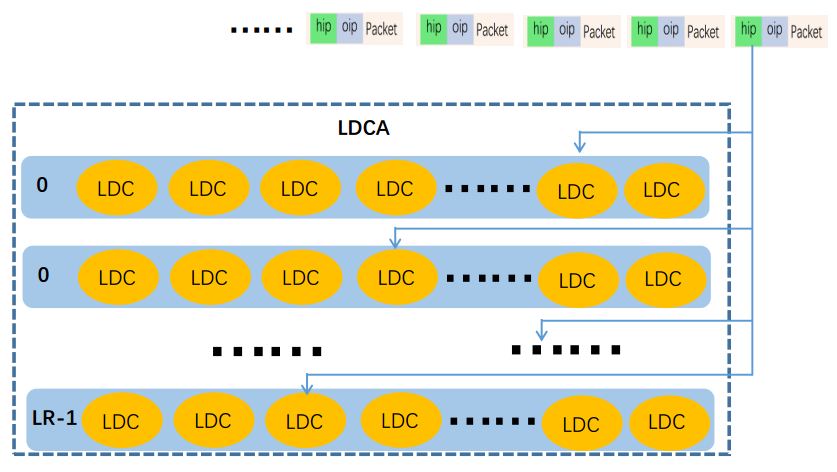}
\caption{Structure of LDCA}
\label{LDCAstruct}
\end{figure}
We use $LR$ random hash functions $LH_i$ to map a host to $LR$ LDCs in each row. No need to restore super points, the updating algorithm is much simpler than SEAV. For a given host at the end of a time window, its opposite hosts number is acquired from the union LDC related with it in each row by equation \ref{eqt_LDC_estValueFromZeroBitsN}. Two LDC is merged by bit-wise ``AND" operation. $LR$ has a great influence on the accuracy of LDCA. The next section we will describe how to set $LR$.

\subsection{Setting the number of rows of LDCA}
The row number of $LDCA$, $LR$, affects the accuracy of cardinality estimation. The larger the value of $LR$, the more SE will be updated when processing each IP pair. In addition to increasing computing time, a large $LR$ may also reduce the accuracy of the algorithm.

When the memory occupied by $LDCA$ is constant, how to set the $LR$ and $LC$ is the most reasonable? This section discusses this issue.

In order to facilitate the discussion, the symbols and parameters are given first. V denotes the number of $LDC$ in $LDCA$, and $V = LR*LC$; $N$ denotes the number of different IP pairs in a time window. For a stable network, the fluctuation range of $N$ in a period should be relatively stable. This value can be obtained from statistical observations of past traffic and can be reset when the observed changes exceed the acceptable range. Therefore, in the following discussion, $N$ is assumed to be a constant.

$Psu$, the probability that a bit in the $LDCA$ is set to ‘1’, is used as the analysis measure. This is because setting $LR$ to a reasonable value can reduce $Psu$ and improve the accuracy of cardinality estimation. In the union $LDC$, $ULDC$, generated by the merging of $LR$ $LDC$, $Psu$ varies with the change of $N$ and $LR$, so $N$ is an unavoidable parameter to discuss this problem.
\begin{lemma}
\label{la-psu_of_different_rowN}
When there are $N$ different IP pairs in a time window, the probability that a bit in $ULDC$ is ‘1’ is $Psu=(1-(1-\frac{1}{k})^{\frac{N}{LC}})^{LR}$
\end{lemma}
\begin{proof}
Let n1 be the number of IP pairs corresponding to a $LDC$. In $LDCA$, each row is updated by $N$ IP pairs. When IP pairs are mapped uniformly to different LDC by hash functions, $n1 = \frac{N}{LC}$. In a $LDC$, the probability that a bit is keep 0 is $(1-\frac{1}{k})^{n1}$. In the merged $ULDC$, a bit is ‘1’ if and only if all the corresponding bit in the $LR$ rows are set to ‘1’, so $Psu=(1-(1-\frac{1}{k})^{n1})^{LR}=(1-(1-\frac{1}{k})^{\frac{N}{LC}})^{LR}$.
\end{proof}
\begin{theorem}
\label{th-different_row_union_le}
If there are a total of V $LDC$ in $LDCA$, i.e. $LR*LC=V$, and $N$ different IP pairs in a time window, then $Psu$ gets the minimum when $LR=\frac{-V*ln(2)}{N*ln(1-\frac{1}{k})}$.
\end{theorem}
\begin{proof}
When $V$ and $N$ are constant，$Psu=(1-(1-\frac{1}{k})^{\frac{LR*N}{V}})^{LR}$. Let $L=\frac{N}{V}$, $G=1-\frac{1}{k}$. The derivative of $Psu$ is $\frac{\mathrm{d}_{Psu} }{\mathrm{d}_{LR}}=(1-G^{LR*L})^{LR}(ln(1-G^{LR*L})-LR*L*G^{LR*L}*ln(G)*(1-G^{LR*L})^{-1})$ 
When $\frac{\mathrm{d}_{Psu} }{\mathrm{d}_{LR}}=0$, $Psu$ reaches its minimum. Since $1-G^{LR*L}>0$, when $\frac{\mathrm{d}_{Psu} }{\mathrm{d}_{LR}}=0$,$ln(1-G^{LR*L})-LR*L*G^{LR*L}*ln(G)*(1-G^{LR*L})^{-1}=0$, and we have the following equation：
\begin{equation}
\label{eq-thr-differRow_410}
ln(1-G^{LR*L})=LR*L*G^{LR*L}*ln(G)*(1-G^{LR*L})^{-1}
\end{equation}
Let $x=LR*L$. Equation \ref{eq-thr-differRow_410} is rewritten as follows:
\begin{equation}
\label{eq-thr-differRow_411}
(1-G^{x})ln(1-G^{x})=G^{x}*ln(G^{x})
\end{equation}
 
When $1-G^x=G^x$, equation \ref{eq-thr-differRow_411} established. Here $G^x=\frac{1}{2}$. Replacing $x$, $G$ and $L$ in equation \ref{eq-thr-differRow_410}, we will get $LR=\frac{-V*ln(2)}{N*ln(1-\frac{1}{k})}$. 
\end{proof}
\textsc{•}For $SLDA$, setting $LR$ reasonably according to the number of IP pairs can improve the accuracy. The $LR$ given in Theorem \ref{th-different_row_union_le} can be used as the upper limit of the number of rows. When $LR$ does not exceed $\frac{-V*ln(2)}{N*ln(1-\frac{1}{k})}$, increasing $LR$ decreases $Psu$. But the larger the $LR$, the more bits need to be merged when generating $ULDC$. This will increase the time consumed by our algorithm. Therefore, the accuracy and calculation time of the algorithm should be taken into account when selecting $LR$ in practice. From Theorem \ref{th-different_row_union_le}, we can see that the larger $V$, the smaller $Psu$, but the larger $LR$. In the actual determination of $LR$, as long as $Psu*k$ is less than 1, that is, the number of noise in the combined $ULDC$ is less than 1, the accuracy of $LDCA$ is as high as $LE$.

LDCA improves the accuracy of detection result. Both $SEAV$ and $LDCA$ can be updated parallel without any data accessing conflict which ensures the success deploying on GPU. In the next section we will introduce how to detect super points in parallel and distributed environment.
\section{Distributed super points detection on GPU}
In a high speed network, such as 40 Gb/s, there are millions of packets passing through the edge of the network. To scan so many packets in real time requires plenty computing resource. Graphic processing unit (GPU) is one of the most popular parallel computing platform in recent years. For these tasks that have no data accessing conflict and processing different data with the same instructions (SIMD), GPU can acquire a high speed up\cite{PD2013:BenchmarkingOfCommunicationTechniquesForGPUs}\cite{PD2013:GeneratingDataTransfersForDistributedGPUParallelPrograms}. Every packet will update SEAV and LDCA. Both these processes just set several bits and every bit could be set by different threads concurrently without introducing any mistakes. Our algorithm has great potential of scanning packets parallel on GPU. When detecting super points on different $WP$s, there are three critical stages: scanning packets on every WP; merging $SEAV$ and $LDCA$ of every WP into global ones; restoring super points from global $SEAV$ and $LDCA$.
\subsection{Packets scanning}
Edge routers transmit packets between SNet and ONet. IP addresses of every packet could be acquired directly at these edge routers. But GPU has its own memory and it can only access its graphic memory directly. So we have to copy IP addresses from WP to GPU's global memory. It's not efficient to copy these IP addresses one by one because the copying processing requires additional starting and ending operations. In order to reducing copying time, we allocate two buffers on WP and GPU separately. When the buffer on WP is full, we will copy it to buffer on GPU and clear it for storing other IP addresses.

GPU has hundreds of cores and can launch thousands of threads to coping with different data parallel. After receiving IP addresses buffer from WP, we will start plenty threads to cope with these IP pairs parallel.

Each WP will only cope with IP pairs of packets passing through it. But a host's opposite hosts may send to different $WP$s. It's not possible to get the accuracy opposite number from one WP. So we should gather all $SEAV$ and $LDCA$ from distributed $WP$s together for global super points restoring.
\subsection{Data merging and super points restoring}
At the end of a time period, we will merge all $SEAV$ and $LDCA$ together. In order to relieve the pressure of WP, we set another server as the global server (GS) to restore super points. All WP will send their $SEAV$ and $LDCA$ to GS. $SEAV$ and $LDCA$ in different $WP$s have the same size and all of them are very small. So communication delay between $WP$s and $GS$ will not cause congestion. 

On GS, we acquired the global $SEAV$ and $LDCA$. Super points will be acquired from global $SEAV$ and opposite hosts number of super points could be calculated from global $LDCA$.

Our algorithm requires small memory and has simple operation, no floating operation. A cheap GPU can acquire a high speed up as showed in our experiment.

\section{Work under sliding time window} \label{sec-slidingWindow}
The previous discussion only focus on super point detection and cardinality estimation under discrete time window. But the result under discrete time window will be affected by the time boundary. An efficient way to solve this problem is to replace discrete time window with sliding time window.
\subsection{Sliding time window \& Discrete time window}
Super point's cardinality estimation under discrete time window is simple because it doesn't need to maintain hosts' state in the previous time slices. But the estimating result has the following two problems:
\begin{enumerate}
\item The result is affected by the starting of a discrete time window, and it fails to detect and estimate the super point which spans the boundary of two adjacent windows.
\item The result is reported with high latency. 
\end{enumerate}
This two weakness of discrete time window comes from its moving step. The moving step of discrete time window equals its size. The bigger monitor period, the higher latency and more errors. Sliding time window solves these two problems together because its moving step has no relation to its window size. But super point cardinality estimation under sliding time window is more complex than that under discrete time window because it maintains hosts state of previous time and estimates super point's cardinality more frequently. 
 Discrete time window and sliding time window are two kinds of the period for cardinality estimating as shown in figure \ref{SlidingDiscreting_time_window}. 
 \begin{figure}[!ht]
\centering
\includegraphics[width=0.47\textwidth]{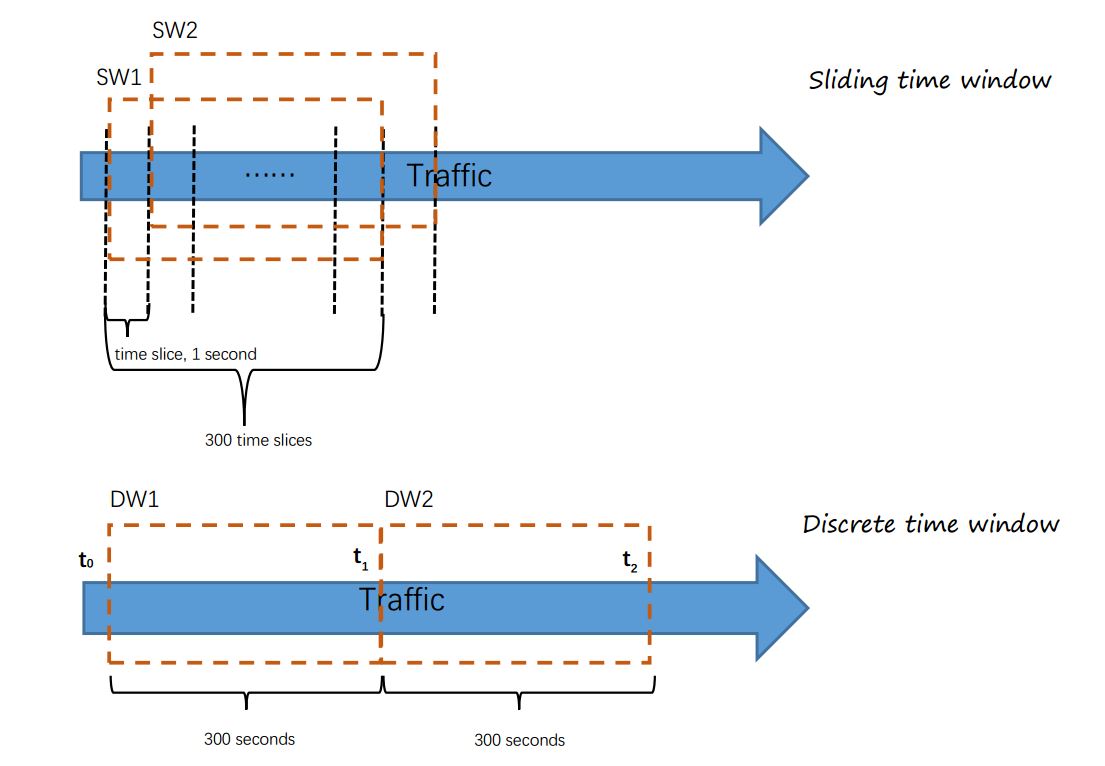}
\caption{Sliding time window and discrete time window}
\label{SlidingDiscreting_time_window}
\end{figure}

Traffic between network SNet and ONet could be divided into successive time slices which have the same duration. The length of a time slice could be 1 second, 1 minute or any period in different situations. A sliding time window, denoted as $W(t, k)$, contains k successive time slices starting from time point t as shown in the top part of figure \ref{SlidingDiscreting_time_window}. Sliding time window will move forward one time slice a time. So two adjacent sliding time windows contain k-1 same slices. When k is set to 1, there is no duplicate period between two adjacent windows, which is the case of the discrete time window in the bottom part of figure \ref{SlidingDiscreting_time_window}. In figure \ref{SlidingDiscreting_time_window}, the size of the time slice is set to 1 second for sliding time window and 300 seconds for the discrete time window. A sliding window in figure \ref{SlidingDiscreting_time_window} contains 300 time slices. In figure \ref{SlidingDiscreting_time_window}, the size of a sliding time window is equal to that of a discrete time window. 

Cardinality estimation under discrete time window is easy because it doesn't need to maintain the appearance of opposite hosts in another time window. But the result is affected by the starting of the discrete time window. When a super point has different opposite hosts in two adjacent time windows, it may be neglected under discrete time window. 

For example, suppose that $DW1$ starts from time point $t_0$ to time point $t_2$ and $DW2$ starts from time point $t_2$ to time point $t_4$ in figure \ref{SlidingDiscreting_time_window}. Let $t_1$ and $t_3$ be two time points in $DW1$ and $DW2$ separately and $TW(t_1, t_3)=300$ seconds. If $|OP(aip, t_1, t_2)|=512$ and $|OP(aip, t_2, t_3)|=512$, $aip$ is a super point in $TW(t_1, t_3)$. But $aip$ will never be detected out in $DW1$ nor $DW2$. By surveying a real-world 40Gb/s network, we found that discrete time window will lose average 14 such super points. We call the super point detected under the sliding window a sliding super point.

\subsection{Sliding super point detection}
In discrete time window, single bit is big enough to store if a host has appeared in a certain period. But under the sliding time window, a counter must store previous state of previous time slices. Hence, a more powerful counter is required. There are many excellent counters used for sliding cardinality estimation, such as time stamp\cite{SDC:MaintainingStreamStatisticsOverSlidingWindows}, Distance Recorder (\textit{DR}) \cite{ISPA2017:HighSpeedNetworkSuperPointsDetectionBasedSlidingWindowGPU} and Asynchronous Timestamp (AT)\cite{VATE2018:sw}. AT has the merits of small memory consumption and few state maintain time at the same time. This paper adopts AT to make SEAV and LDCA run under sliding window.

Under the sliding window, each bit in SEAV or LDCA is replaced by an AT. We call the discrete time window version SLGA and the sliding version of super point detection algorithm SSLGA. SLAG needs to initialize SEAV and LDCA at the beginning of each time window as described in the previous section. SSLGA only needs to initialize once at the beginning of the algorithm. Then SSLGA updates AT incrementally. 

Because the number of bits in SEAV and LDCA are constant while running, the number of AT in SSLGA are constant too while running. A pool containing fixed number of AT is allocated at the beginning of SSLGA. Each AT in SSLGA corresponds to an AT in the AT pool. Instead of reinitialization at the beginning of each time window, SSLGA maintains the states of its AT at the end of each time slice. The state maintaining of SSLGA is to maintain the state of AT in the pool. The AT in the pool is divided according to the method in paper \cite{VATE2018:sw}. The maintaining time of AT is small. Hence SSLGA has a fast speed. SSLGA detects super points like SLGA. The inactive AT corresponds to ``0" bit and the active AT corresponds to ``1" bit. AT has the ability to run parallel. Hence SSLGA could be deployed on GPU to running in real time.

\section{Experiment and analysis} \label{sec-experiments}
In order to evaluate the performance of our algorithm, we use six real world core network traffic to compare the accuracy and consumption time of our algorithm with others. These traffics could be downloaded from Caida\cite{expdata:Caida}. Caida's OC192 data monitors an hour-long trace starting from 13:00. Table \ref{tbl-trafficInf} shows the detail information of every experiment traffic.

\begin{table*}
\centering
\caption{Traffic information}
\label{tbl-trafficInf}
\begin{tabular}{c}                                                                                                                                                                                                                           
\centering
\includegraphics[width=0.95\textwidth]{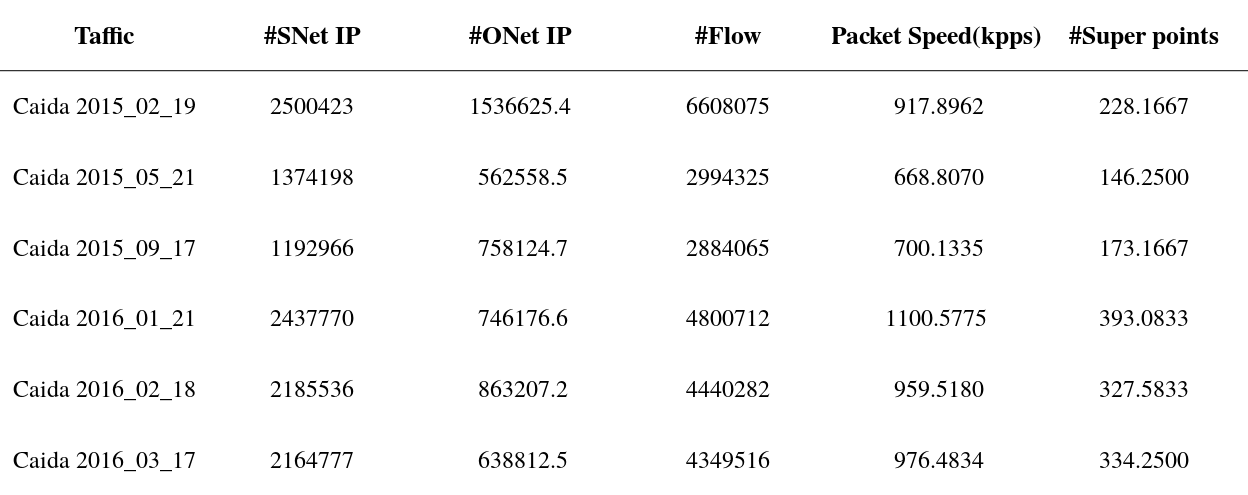}
\end{tabular}
\end{table*}

We use a common and low cost GPU, Nvidia GTX 650, to run every algorithm. There are total 1 GB of graphic memory in this GPU card and it communicates with computer through PCIe3.0 which has a bandwidth as high as 16 Gb/s. We compare the performance of different algorithms: DCDS\cite{HSD:ADataStreamingMethodMonitorHostConnectionDegreeHighSpeed}, VBFA \cite{HSD:DetectionSuperpointsVectorBloomFilter}, GSE \cite{HSD:GPU:2014:AGrandSpreadEstimatorUsingGPU} and SLGA. SLGA is the one proposed in this paper.
\subsection{Accuracy and memory}
False positive rate (FPR) and false negative rate (FNR) are two important criteria of detection accuracy. FPR means the ratio of the number of detected fake host to the number of super points. FNR is the ration of the number of these super points that failed be detected by an algorithm to the number of super points. For an estimating algorithm, we hope that its FPR and FNR are small at the same time because FPR shows a negative correlation with FNR. Figure \ref{Rlt_FPR} and \ref{Rlt_FNR} illustrate the FPR and FNR of different algorithms. In our algorithm, the time window is set to 300 seconds and each one-hour traffic is split into 12 sub traffics according to the time window.

\begin{figure*}[!ht]
\centering
\includegraphics[width=0.97\textwidth]{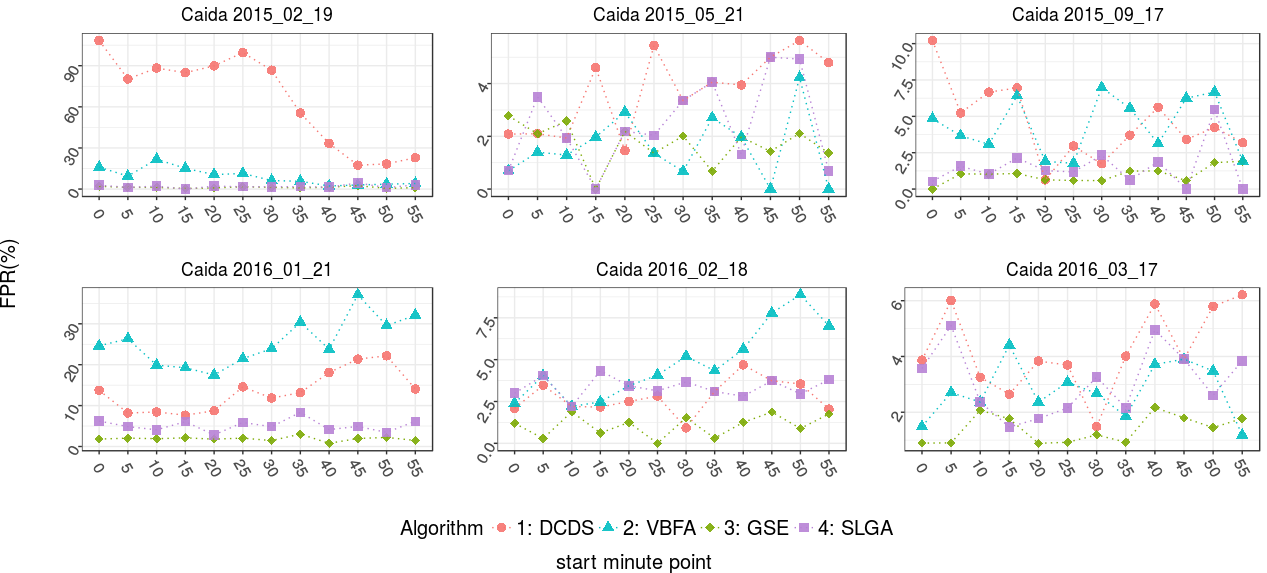}
\caption{FPR of different algorithms}
\label{Rlt_FPR}
\end{figure*}

\begin{figure*}[!ht]
\centering
\includegraphics[width=0.97\textwidth]{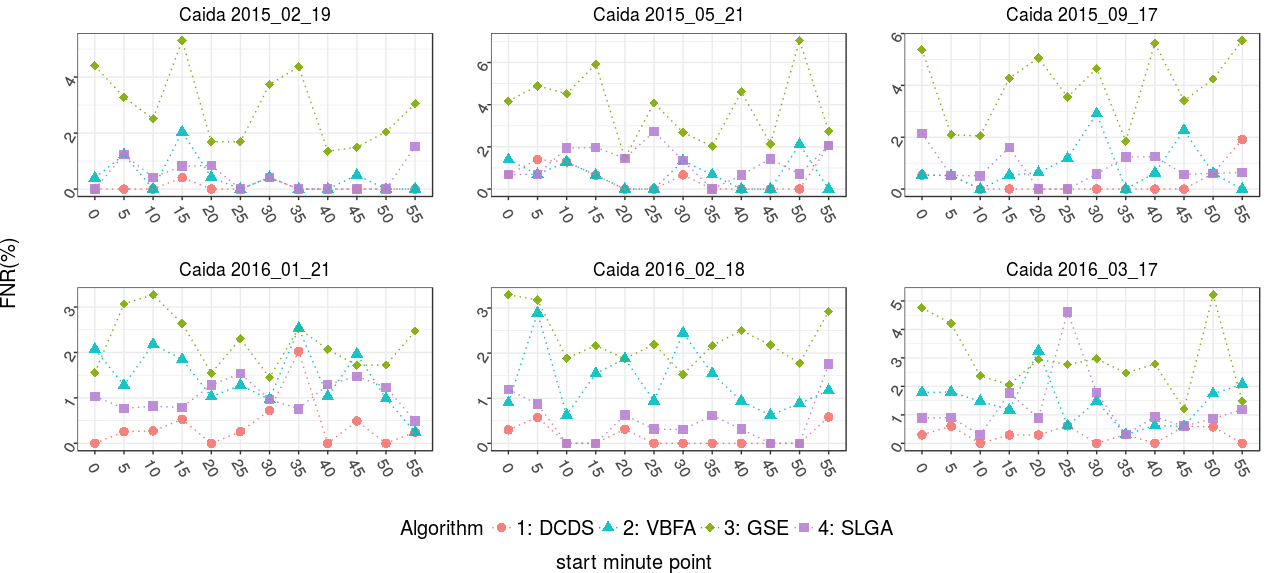}
\caption{FNR of different algorithms}
\label{Rlt_FNR}
\end{figure*}

GSE has a smaller FPR than other algorithms have. SLGA's FPR is a little higher than that of GSE but much smaller than that of DCDS and VBFA.

Although GSE has a small FPR, its FNR is higher than other algorithms'. A high FNR will let GSE fail to detect some important super points. DCDS's FNR is the lowest at the cost of its high FPR. SLGA's FNR is between DCDS's and VBFA's. In order to have an overall detection accuracy, we use the sum of FPR and FNR, called false total rate (FTR), as the accuracy criterion. SLGA's FTR is the smallest in all of these algorithm. Table \ref{Table_alg_avgRlt} shows the memory consumption and average result of different algorithms.

\begin{table*}
\centering
\caption{Average result of different algorithms}
\label{Table_alg_avgRlt}
\begin{tabular}{c}                                                                                                                                                                                                                           
\centering
\includegraphics[width=0.95\textwidth]{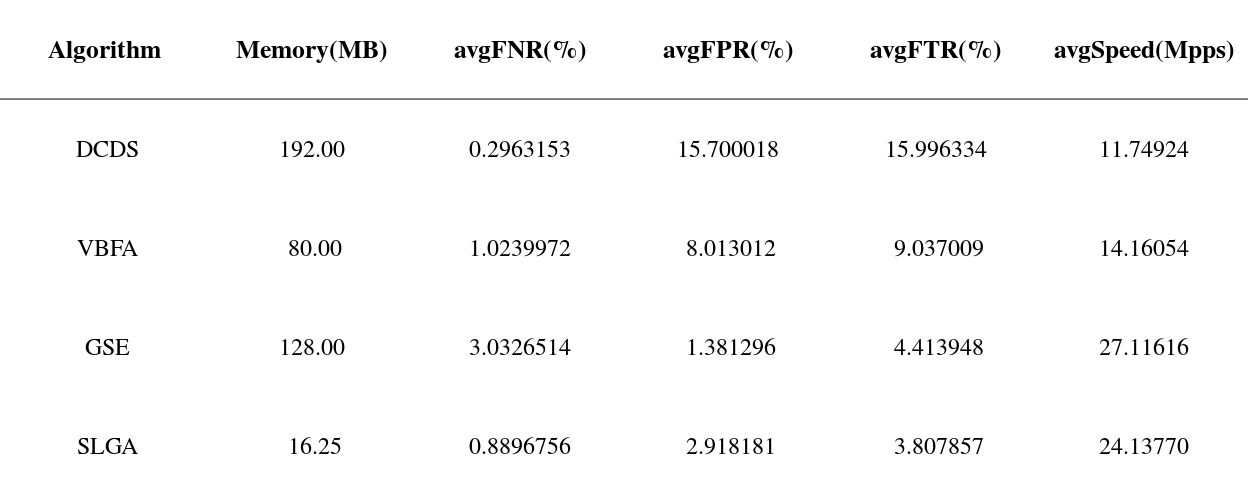}
\end{tabular}
\end{table*}
Both DCDS and GSE use more than 100 MB memory and DCDS uses nearly 200 MB memory. VBFA uses smaller memory than DCDS and GSE. But SLGA consumes the smallest memory in all of these algorithm, only one-fifth memory that VBFA uses. Small memory requirement let SLGA has a small communication latency in a distributed environment than other algorithms. 
False rates listing in table \ref{Table_alg_avgRlt} are acquired by calculating the average value of an algorithm at all these time windows. From table \ref{Table_alg_avgRlt} we can see that SLGA not only has the smallest memory requirement, but also the smallest overall false rate. Its processing speed is fast enough for deal with this traffic in real time.
\subsection{Time consumption}
When running in our cheap GPU, all algorithm can detect super points in real time for every 5-minutes sub traffic. But their consuming time are very different as shown in figure \ref{Rlt_TotalUT}.
\begin{figure*}[!ht]
\centering
\includegraphics[width=0.97\textwidth]{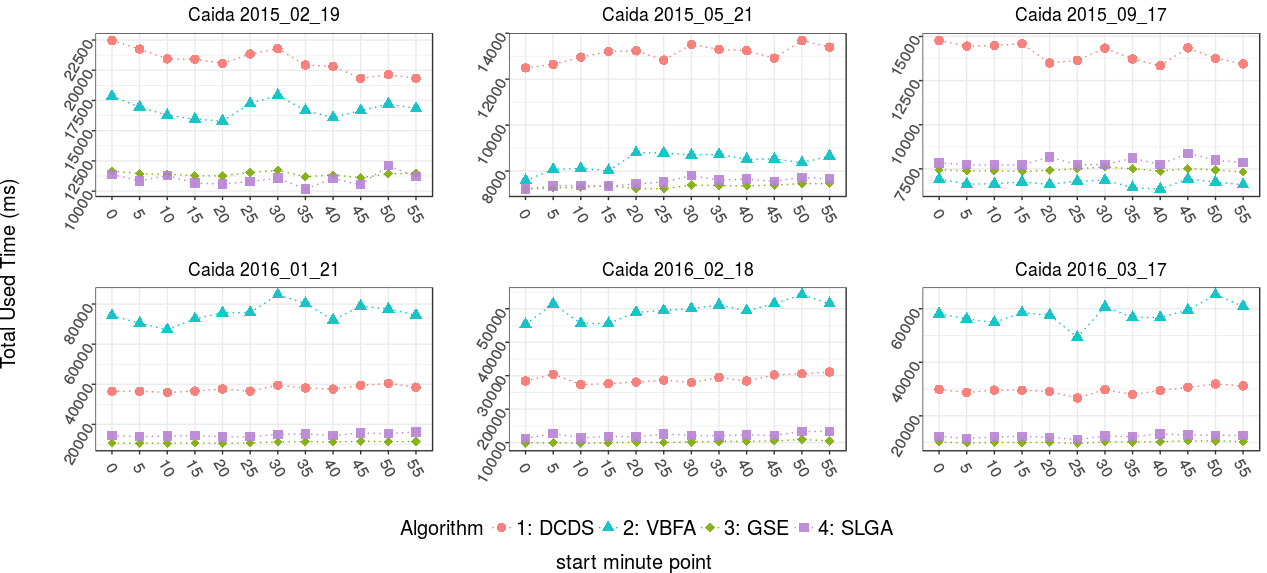}
\caption{Total used time of different algorithms for every sub traffic}
\label{Rlt_TotalUT}
\end{figure*}

DCDS uses much more time than GSE and SLGA because it employs CRT to restore super points which requires more complex operation than a hash function does. In the first three traffics, VBFA uses littler time than DCDS does. This is because VBFA can locate column index by extracting sub bits of hosts's IP addresses. But when reconstructing super points, VBFA will generate huge candidate IP addresses. The number of candidate IP addresses increases sharply with the number of super points. In the last three traffic, which all contain more than 300 super points, VBFA uses much more time than all the other algorithms.

SLGA uses a little more time than GSE. When scanning traffic, GSE only needs to set one bit while SLGA will set several bits in $SEA$ and $LDCA$. But GPU can hide the memory accessing delay by launching plenty threads parallel. So SLGA is slower than GSE a little. In this paper, we divide the total packets number in a time window by the processing time to get the speed. The unit of algorithm's speed is million packets per seconds, written as Mpps. Speeds of different algorithms listed in table \ref{Table_alg_avgRlt} is the mean value of different algorithms's speeds of all 5-minutes sub traffics.

SLGA and GSE have faster speed than VBFA and DCDS. The experiment traffics have an average bandwidth of 4.5 Gb/s \cite{expdata:Caida}. Suppose that 900 MB memory of GPU are available for different algorithms' kernel structures and the rest 100 MB memory are used for IP addresses buffer and other running parameters. For a higher bandwidth traffic, more memory will be required. Suppose that the memory requirement grows linearly with the bandwidth which could be realized by splitting IP addresses by their right bits. From the memory perspective, our algorithm can deal with network traffic with bandwidth as high as 249 Gb/s, while the highest traffic bandwidth for DCDS, VBFA and GSE are 21.1 Gb/s, 50.6 Gb/s and 31.6 Gb/s. Note that the GPU used in our experiment is a cheap one which costs only 30 dollars. A more advanced GPU, such as GTX 1080 with 11 GB graphic memory, could be brought with 1000 dollars to deal with a faster and bigger network.

\section{Conclusion} \label{sec-conclusion}
In this paper we introduce a memory efficient distributed super points detection algorithm. Super point plays important roles in network management and security. How to find them out in real time is the foundation of super points application. Unlike other algorithms, we use two kinds of opposite number estimation algorithms in our scheme: short estimation and long estimation. SE consumes very small memory and has a fast processing speed. Based on SE, we design a novel super point restoring structure SEAV. From SEAV we can get a candidate super points lists. In order to improve the accuracy of the detection result, we introduce LE. LE consumes more memory than SE but has a higher accuracy. Using SE and LE together makes our algorithm get the highest accuracy with the smallest memory.

\bibliographystyle{splncs04}

\bibliography{..//..//ref} 

\begin{thebibliography}{10}
\providecommand{\url}[1]{\texttt{#1}}
\providecommand{\urlprefix}{URL }
\providecommand{\doi}[1]{https://doi.org/#1}

\bibitem{expdata:Caida}
for Applied Internet Data~Analysis, C.: The caida anonymized internet traces.
  \url{http://www.caida.org/data/passive} (2017), online;accessed 2017

\bibitem{PD2013:BenchmarkingOfCommunicationTechniquesForGPUs}
Bernaschi, M., Bisson, M., Rossetti, D.: Benchmarking of communication
  techniques for gpus. Journal of Parallel and Distributed Computing
  \textbf{73}(2),  250 -- 255 (2013).
  \doi{https://doi.org/10.1016/j.jpdc.2012.09.006},
  \url{http://www.sciencedirect.com/science/article/pii/S0743731512002213}

\bibitem{scan:surveyPortScansAndDetection}
Bhuyan, M.H., Bhattacharyya, D., Kalita, J.: Surveying port scans and their
  detection methodologies. Comput. J.  \textbf{54}(10),  1565--1581 (Oct 2011).
  \doi{10.1093/comjnl/bxr035}, \url{http://dx.doi.org/10.1093/comjnl/bxr035}

\bibitem{HSD:identifyHighCardinalityHosts}
Cao, J., Jin, Y., Chen, A., Bu, T., Zhang, Z.L.: Identifying high cardinality
  internet hosts. In: IEEE INFOCOM 2009. pp. 810--818 (April 2009).
  \doi{10.1109/INFCOM.2009.5061990}

\bibitem{hash_UniversalClassesOfHashFunctions}
Carter, J., Wegman, M.N.: Universal classes of hash functions. Journal of
  Computer and System Sciences  \textbf{18}(2),  143 -- 154 (1979).
  \doi{http://dx.doi.org/10.1016/0022-0000(79)90044-8},
  \url{http://www.sciencedirect.com/science/article/pii/0022000079900448}

\bibitem{cisco:NetForcast}
Cisco: Global ip traffic forecast.
  \url{http://www.cisco.com/c/en/us/solutions/collateral/service-provider/visual-networking-index-vni/vni-hyperconnectivity-wp.pdf}
  (2017), online

\bibitem{SDC:MaintainingStreamStatisticsOverSlidingWindows}
Datar, M., Gionis, A., Indyk, P., Motwani, R.: Maintaining stream statistics
  over sliding windows. SIAM J. Comput.  \textbf{31}(6),  1794--1813 (Jun
  2002). \doi{10.1137/S0097539701398363},
  \url{https://doi.org/10.1137/S0097539701398363}

\bibitem{HSD:bitmapCountingActiveFlowsHighSpeedLinks}
Estan, C., Varghese, G., Fisk, M.: Bitmap algorithms for counting active flows
  on high-speed links. IEEE/ACM Trans. Netw.  \textbf{14}(5),  925--937 (Oct
  2006). \doi{10.1109/TNET.2006.882836},
  \url{http://dx.doi.org/10.1109/TNET.2006.882836}

\bibitem{Scan:EvasionResistantNetworkScanDetection}
Harang, R.E., Mell, P.: Evasion-resistant network scan detection. Security
  Informatics  \textbf{4}(1), ~4 (2015). \doi{10.1186/s13388-015-0019-7},
  \url{http://dx.doi.org/10.1186/s13388-015-0019-7}

\bibitem{IMC:MeasuringtheAdoptionofDDoSProtectionServices}
Jonker, M., Sperotto, A., van Rijswijk-Deij, R., Sadre, R., Pras, A.: Measuring
  the adoption of ddos protection services. In: Proceedings of the 2016
  Internet Measurement Conference. pp. 279--285. IMC '16, ACM, New York, NY,
  USA (2016). \doi{10.1145/2987443.2987487},
  \url{http://doi.acm.org/10.1145/2987443.2987487}

\bibitem{DC:AnOptimalAlgorithmDistinctElementProblem}
Kane, D.M., Nelson, J., Woodruff, D.P.: An optimal algorithm for the distinct
  elements problem. In: Proceedings of the Twenty-ninth ACM
  SIGMOD-SIGACT-SIGART Symposium on Principles of Database Systems. pp. 41--52.
  PODS '10, ACM, New York, NY, USA (2010). \doi{10.1145/1807085.1807094},
  \url{http://doi.acm.org/10.1145/1807085.1807094}

\bibitem{ACSAC:2014:DoS:CPSDrivingCyberphysicalSystems}
Krotofil, M., C\'{a}rdenas, A.A., Manning, B., Larsen, J.: Cps: Driving
  cyber-physical systems to unsafe operating conditions by timing dos attacks
  on sensor signals. In: Proceedings of the 30th Annual Computer Security
  Applications Conference. pp. 146--155. ACSAC '14, ACM, New York, NY, USA
  (2014). \doi{10.1145/2664243.2664290},
  \url{http://doi.acm.org/10.1145/2664243.2664290}

\bibitem{HSD:DetectionSuperpointsVectorBloomFilter}
Liu, W., Qu, W., Gong, J., Li, K.: Detection of superpoints using a vector
  bloom filter. IEEE Transactions on Information Forensics and Security
  \textbf{11}(3),  514--527 (March 2016). \doi{10.1109/TIFS.2015.2503269}

\bibitem{HSD:IdentifyHighCardinalitHostNetworkWideTrafficMeasurement}
Liu, Y., Chen, W., Guan, Y.: Identifying high-cardinality hosts from
  network-wide traffic measurements. IEEE Transactions on Dependable and Secure
  Computing  \textbf{13}(5),  547--558 (Sept 2016).
  \doi{10.1109/TDSC.2015.2423675}

\bibitem{ICPADS2011:AFailureDetectionServiceForInternetBasedMultiASDistributedSystems}
Moraes, D.M., Jr., E.P.D.: A failure detection service for internet-based
  multi-as distributed systems. In: 2011 IEEE 17th International Conference on
  Parallel and Distributed Systems. pp. 260--267 (Dec 2011).
  \doi{10.1109/ICPADS.2011.5}

\bibitem{Secure:SnortLightweightIntrusionDetectionNetworks}
Roesch, M.: Snort - lightweight intrusion detection for networks. In:
  Proceedings of the 13th USENIX Conference on System Administration. pp.
  229--238. LISA '99, USENIX Association, Berkeley, CA, USA (1999),
  \url{http://dl.acm.org/citation.cfm?id=1039834.1039864}

\bibitem{P2P:SoKP2PWNED}
Rossow, C., Andriesse, D., Werner, T., Stone-Gross, B., Plohmann, D., Dietrich,
  C.J., Bos, H.: Sok: P2pwned - modeling and evaluating the resilience of
  peer-to-peer botnets. In: 2013 IEEE Symposium on Security and Privacy. pp.
  97--111 (May 2013). \doi{10.1109/SP.2013.17}

\bibitem{HSD:GPU:2014:AGrandSpreadEstimatorUsingGPU}
Shin, S.H., Im, E.J., Yoon, M.: A grand spread estimator using a graphics
  processing unit. Journal of Parallel and Distributed Computing
  \textbf{74}(2),  2039 -- 2047 (2014).
  \doi{http://dx.doi.org/10.1016/j.jpdc.2013.10.007},
  \url{http://www.sciencedirect.com/science/article/pii/S0743731513002189}

\bibitem{PD2013:GeneratingDataTransfersForDistributedGPUParallelPrograms}
Silber-Chaussumier, F., Muller, A., Habel, R.: Generating data transfers for
  distributed gpu parallel programs. Journal of Parallel and Distributed
  Computing  \textbf{73}(12),  1649 -- 1660 (2013).
  \doi{https://doi.org/10.1016/j.jpdc.2013.07.022},
  \url{http://www.sciencedirect.com/science/article/pii/S0743731513001603},
  heterogeneity in Parallel and Distributed Computing

\bibitem{IMC:BrowserFeatureUsageModernWeb}
Snyder, P., Ansari, L., Taylor, C., Kanich, C.: Browser feature usage on the
  modern web. In: Proceedings of the 2016 Internet Measurement Conference. pp.
  97--110. IMC '16, ACM, New York, NY, USA (2016).
  \doi{10.1145/2987443.2987466},
  \url{http://doi.acm.org/10.1145/2987443.2987466}

\bibitem{HSD:streamingAlgorithmFastDetectionSuperspreaders}
Venkataraman, S., Song, D., Gibbons, P.B., Blum, A.: New streaming algorithms
  for fast detection of superspreaders. In: in Proceedings of Network and
  Distributed System Security Symposium (NDSS. pp. 149--166 (2005)

\bibitem{DosC:AttackProtectionCloudComputingSoftwareDefinedNetworking}
Wang, B., Zheng, Y., Lou, W., Hou, Y.T.: \{DDoS\} attack protection in the era
  of cloud computing and software-defined networking. Computer Networks
  \textbf{81},  308 -- 319 (2015).
  \doi{http://doi.org/10.1016/j.comnet.2015.02.026},
  \url{http://www.sciencedirect.com/science/article/pii/S1389128615000742}

\bibitem{HSD:ADataStreamingMethodMonitorHostConnectionDegreeHighSpeed}
Wang, P., Guan, X., Qin, T., Huang, Q.: A data streaming method for monitoring
  host connection degrees of high-speed links. IEEE Transactions on Information
  Forensics and Security  \textbf{6}(3),  1086--1098 (Sept 2011).
  \doi{10.1109/TIFS.2011.2123094}

\bibitem{DC:aLinearTimeProbabilisticCountingDatabaseApp}
Whang, K.Y., Vander-Zanden, B.T., Taylor, H.M.: A linear-time probabilistic
  counting algorithm for database applications. ACM Trans. Database Syst.
  \textbf{15}(2),  208--229 (Jun 1990). \doi{10.1145/78922.78925},
  \url{http://doi.acm.org/10.1145/78922.78925}

\bibitem{DosC:AgainstDataCenterWithCorrelationAnalysis}
Xiao, P., Qu, W., Qi, H., Li, Z.: Detecting \{DDoS\} attacks against data
  center with correlation analysis. Computer Communications  \textbf{67},  66
  -- 74 (2015). \doi{http://doi.org/10.1016/j.comcom.2015.06.012},
  \url{http://www.sciencedirect.com/science/article/pii/S0140366415002285}

\bibitem{ISPA2017:HighSpeedNetworkSuperPointsDetectionBasedSlidingWindowGPU}
Xu, J., Ding, W., Gong, J., Hu, X., Liu, J.: High speed network super points
  detection based on sliding time window by gpu. In: 2017 IEEE International
  Symposium on Parallel and Distributed Processing with Applications and 2017
  IEEE International Conference on Ubiquitous Computing and Communications
  (ISPA/IUCC). pp. 566--573 (Dec 2017). \doi{10.1109/ISPA/IUCC.2017.00092}

\bibitem{VATE2018:sw}
Xu, J., Ding, W., Hu, X.: {VATE:} a trade-off between memory and preserving
  time for high accuracy cardinalities estimation under sliding time window.
  CoRR  \textbf{abs/1812.00282} (2018), \url{http://arxiv.org/abs/1812.00282}

\end{thebibliography}

\end{document}